\documentclass[preprint,12pt]{elsarticle}

\usepackage{amssymb}

\usepackage{graphics} 
\usepackage{epsfig} 
\usepackage{mathptmx} 
\usepackage{times} 
\usepackage{amsthm}
\usepackage{amsmath} 
\usepackage{amssymb, color} 
\usepackage{stmaryrd}
\newtheorem{theo}{Theorem}
\newtheorem{lem}{Lemma}

\newtheorem{assert}{Assertion}
\newtheorem{defi}{Definition}

\newtheorem{remark}{Remark}
\newtheorem*{pb}{Problem}

\newcommand{\identity}{I}
\newcommand{\rref}[1]{(\ref{#1})}

\newcommand{\scal}{\mathcal{S}}

\newcommand{\pcal}{\mathcal{P}}
\newcommand{\norme}[1]{\left\Vert #1\right\Vert}
\newcommand{\abs}[1]{\vert #1 \vert}
\newcommand{\reels}{\mathbb{R}}

\newcommand{\entiers}{\mathbb{N}}
\newcommand{\sign}{\text{sign}}

\journal{Systems $\&$ Controls Letters}

\begin{document}

\begin{frontmatter}



\title{Global stabilization of classes of linear control systems with bounds on the feedback and its successive derivatives\tnoteref{tit1}}

\tnotetext[tit1]{This research was partially supported by a public grant overseen by the French ANR as part of the “Investissements d'Avenir” program, through the iCODE institute, research project funded by the IDEX Paris-Saclay, ANR-11-IDEX-0003-02.}

\author{Jonathan Laporte, Antoine Chaillet and Yacine Chitour}
\ead{jonathan.laporte, antoine.chaillet, yacine.chitour@l2s.centralesupelec.fr}
\address{L2S - Univ. Paris Sud - CentraleSup\' elec. 3, rue Joliot-Curie. 91192 - Gif sur Yvette, France}

\begin{abstract}
In this paper, we address the problem of globally stabilizing a linear time-invariant (LTI) system by means of a static feedback law whose amplitude and successive time derivatives, up to a prescribed order $p$, are bounded by arbitrary prescribed values. We solve this problem for two classes of LTI systems, namely integrator chains and skew-symmetric systems with single input. For the integrator chains, the solution we propose is based on the nested saturations introduced by A.R. Teel. We show that this construction fails for skew-symmetric systems and propose an alternative feedback law. We illustrate these findings by the stabilization of the third order integrator with prescribed bounds on the feedback and its first two derivatives, and similarly for the harmonic oscillator with prescribed bounds on the feedback and its first derivative.

\end{abstract}

\begin{keyword}
Global stabilization \sep Bounded controls \sep Rate saturation \sep Multiple integrators \sep Oscillators.



\end{keyword}

\end{frontmatter}


\section{Introduction}
\label{Intro}

Actuator constraints constitute an important practical issue in control applications since they are a possible source of instability or performance degradation. Strong research efforts have been devoted to the stabilization of linear time-invariant (LTI) plants. LTI systems are known to be global stabilizable despite actuator saturations (i.e., by bounded inputs) if and only if they are stabilizable in the absence of input constraints and their internal dynamics has no eigenvalues with positive real part \cite{SSY}. For systems that do not fulfill these constraints, several approaches provide stabilization from an arbitrarily large compact set of initial conditions (semiglobal stability); these include for instance \cite{Lin1993225,AlvarezRamirez1994247,hu2001semi,Tarbouriech:2006ew}. Some of these semiglobal approaches can be extended to robust stabilization in presence of exogenous disturbances \cite{Saberi:2002ux,Dai:2009kv}.

The objective of globally stabilizing plants by a bounded feedback remains of practical relevance, since the resulting control gains do not depend on the magnitude of initial states. Procedures ensuring global stability of nonlinear plants by bounded feedback have been proposed in \cite{Lin:1991vw,Mazenc:2002wt} and robustness to exogenous inputs have been addressed in \cite{ANCHMA05,AZCHCHGR15}. Among the LTI systems that can be globally stabilized by bounded feedback, chains of integrators have received specific attention. The simple saturation of a linear feedback fails at ensuring global stability as soon as the integrator chain is of dimension greater than or equal to three \cite{FULLER69,SY91}. In \cite{Teel92} a globally stabilizing feedback was constructed using {\it nested saturations} for a chain of integrators of arbitrary length. This construction has been extended to all LTI plants that can be stabilized by bounded feedback in \cite{SSY}, in which a family of stabilizing feedback laws was proposed as a linear combination of saturation functions.  In \cite{marchand2005global}, the issue of performance of these bounded feedbacks is investigated for chains of integrators and some improvements are achieved by using variable levels of saturation. Global approaches ensuring robustness to exogenous disturbances have also been investigated.  The first general solution to the $L_p$ finite-gain stabilization problem was provided in \cite{saberi2000}, based on a gain scheduled feedback initially proposed in \cite{Megretski96bibooutput}. An alternative easily implementable solution to that problem was recently proposed in \cite{chitour2015} for chains of integrators. 
As for neutrally stable systems, such a solution was first given in \cite{LICHSO96}.

While actuation magnitude is often the main concern in practical applications, limited actuation reactivity can also be an issue. Indeed, technological constraints may affect not only the amplitude of the delivered control signal, but also the amplitude of its time derivative. This latter problem is known as rate saturation and has been addressed for instance in \cite{lauvdal97,Freeman:1998tp, SilvaTarbouch03,Galeani,saberi2012}.  Semiglobal stabilization has been achieved via a gain scheduling technique \cite{lauvdal97}, or through low-gain feedback or low-and-high-gain feedback \cite{saberi2012}. In \cite{SilvaTarbouch03, Galeani}, regional stability was ensured through LMI-based conditions. In \cite{Freeman:1998tp}, this problem has been addressed for nonlinear plants using backstepping procedure ensuring global stability. 

In this paper, we deepen the investigations on global stabilization of integrator chains subject to bounded actuation with rate constraints. We consider rate constraints that affect not only the first time derivative of the control signal, but also its successive $p$ first derivatives, where $p$ denotes an arbitrary positive integer. We specifically study two classes of systems that can be globally stabilized by bounded state feedback, namely chains of integrator and skew-symmetric dynamics with single input.  No restriction is imposed on the dimension of these systems. We show that solving the problem for these two cases actually cover wider classes of systems, namely all systems with either only zero eigenvalues or only simple eigenvalues with zero real part. For both these classes of LTI systems, we propose a bounded static feedback law that ensures global asymptotic stability of the closed-loop system, and whose magnitude and $p$ first time derivatives are bounded by arbitrary prescribed values. For the chains of integrators, the proposed control law is based on the nested saturations procedure introduced in \cite{Teel92}. We rely on specific saturation functions, which are linear in a neighborhood of the origin and constant for large values of their argument. Unfortunately, we show that this nested saturations feedback fails solve the problem for skew-symmetric dynamics. For the latter class of systems, we propose an alternative construction.

This paper is organized as follows. In Section \ref{section:stat_main_res}, we provide definitions and state our main results for both considered classes of LTI systems. The proofs of the main results are provided in Section \ref{sec: proof} based on several technical lemmas. In Section \ref{sec:simu11}, we test the efficiency of the proposed control laws via numerical simulations on the third order integrator and the harmonic oscillator, where we bound with prescribed values the feedback, as well as its first two time derivatives for the  third order integrator and it first time derivative for the harmonic oscillator respectively.

\textbf{Notations.}
The function $\sign:\reels \backslash \lbrace{0\rbrace}\to \mathbb R$ is defined as $\sign(r) := r / \abs{r}$. Given a set $I\subset \mathbb R$ and a constant $a\in\mathbb R$, we let $I_{\geq a}:=\left\{x\in I\,:\, x\geq a\right\}$.  Given $k\in\mathbb N$ and $n,p\in\mathbb N_{\geq 1}$, we say that a function $f : \reels^n \rightarrow \reels^p$ is of class $C^{k}(\reels^n , \reels^p)$ if its differentials up to order $k$ exist and are continuous, and we use $f^{(k)}$ to denote the $k$-th order differential of $f$. By convention, $f^{(0)}:=f$. The factorial of $k$ is denoted by $k!$ and the binomial coefficient is denoted $\binom{k}{m}:=\frac{k!}{m!(k-m)!}$. We define $\llbracket m , k \rrbracket:=\left\{l\in\mathbb N\,:\, l\in[m,k]\right\}$. We use $\reels^{n,n}$ to denote the set of $n\times n$ matrices with real coefficients. The matrices $\identity_n$ and $J_n\in \reels^{n,n}$ denote the identity matrix of dimension $n$ and the $n$-th Jordan block respectively, i.e., the $n\times n$ matrix given by $(J_n)_{i,j} =1$ if $i=j-1$ and zero otherwise.
For each $i\in\llbracket 1 , n \rrbracket$, $e_i\in \reels^n$ refers to the column vector with coordinates equal to zero except the $i$-th one equal to one. We use $\Vert x\Vert$ to denote the Euclidean  norm of an arbitrary vector $x\in\reels^n$. For two sets $A$ and $B$, the relative complement of $A$ in $B$ is denoted by $B \backslash A$.
\section{Statement of the main results}
\label{section:stat_main_res}

\subsection{Problem statement}

We start by introducing in more details the general problem we address. Given $n\in\mathbb N_{\geq 1}$,  consider LTI systems with single input:
\begin{equation}
\label{sys:linear}
\dot{x}= Ax +Bu,
\end{equation}
where $x \in \reels^n$, $A$ and $B$ are $n \times n$ and $n \times 1$ matrices respectively. Assume that the pair $(A,B)$ is stabilizable and that all the eigenvalues of $A$ have non positive real parts. Recall that these assumptions on $(A,B)$ are necessary and sufficient for the existence of a bounded continuous state feedback $u=k(x)$ which globally asymptotically stabilizes the closed-loop system \cite{SSY}. We say that an eigenvalue of $A$ is {\it critical} if it has zero real part. 

Given a family of prescribed bounds $(R_j)_{0 \leq j \leq p}$ on the control signal and its successive $p$-first derivative, we start by introducing the notion of $p$-bounded feedback law by $(R_j)_{0 \leq j \leq p}$ for system \rref{sys:linear}. This terminology will be used all along the document.

\begin{defi}
Given $n\in\mathbb N_{\geq 1}$ and $p\in\mathbb N$, let $(R_j)_{0 \leq j \leq p}$ denote a family of positive constants. We say that $\nu : \reels^n \rightarrow \reels $ is a \textit{$p$-bounded feedback law by $(R_j)_{0 \leq j \leq p}$  for system \rref{sys:linear}} if it is of class $C^p(\reels^n,\reels)$ and, for every trajectory of the closed-loop system $\dot{x}= A x +B \nu(x)$, the control signal $U : \reels_{\geq 0} \rightarrow  \reels $ defined by $U(t) := \nu (x(t)) $ for all $t \geq 0$ satisfies, for all $j \in \llbracket 1 , p \rrbracket$, $\sup \lbrace  \abs{U^{(j)}(t)} : t \geq 0 \rbrace  \leq R_j $.
\end{defi}

Based on this definition, we can restate our stabilization problem as follows.
\begin{pb}
Given $p\in\mathbb N$ and a family of positive real numbers $(R_j)_{0 \leq j \leq p}$, design a feedback law $\nu $ such that the origin of the closed-loop system $\dot{x}=  A x +B \nu(x)$ is globally asymptotically stable (GAS for short) and the feedback $\nu : \reels^n \rightarrow \reels $ is a $p$-bounded feedback law by $(R_j)_{0 \leq j \leq p}$ for System \rref{sys:linear}.
\end{pb}

The case $p=0$ corresponds to global stabilization with bounded state feedback and has been addressed in e.g. \cite{Teel92,SSY}. The case $p=1$ corresponds to global stabilization with bounded state feedback and limited rate, in the line of e.g. \cite{SilvaTarbouch03, Galeani, lauvdal97, saberi2012, Freeman:1998tp}.

In this paper, we present a general solution to the problem at stake for two classes of LTI systems: \textbf{Case 1} all the critical eigenvalues of $A$ are zero, \textbf{Case 2}  all the critical eigenvalues of A are simple and have zero real parts.

Since the pair $(A,B)$ is stabilizable, there exists a linear change of coordinates transforming the matrices $A$ and $B$ into $
\begin{pmatrix}A_1\  A_3\\0\ \ A_2\end{pmatrix}$ and $ \begin{pmatrix}B_1\\ B_2\end{pmatrix}$, 
where $A_1$ is Hurwitz, the eigenvalues of $A_2$ have zero real parts and the pair $(A_2,B_2)$ is controllable. Then, it is immediate to see that we only have to treat the case where $A$ has only critical eigenvalues. From now on, we therefore assume that $A$ has only eigenvalues with zero real part, and that the pair $(A,B)$ is controllable.

\subsection{Multiple integrators}
In Case 1, up to a linear change of coordinates, $A$ can be put in a block-diagonal form with Jordan blocks $J_r$ on the diagonal. It is then clear that, up to an additional linear change of coordinates, addressing Case 1 amounts to dealing with the sole case of a multiple integrator of arbitrary length $n$, i.e. the LTI control system given by
\begin{equation}
\label{mult_int}
\dot{x}_i = x_{i+1}, \quad i < n \quad  \text{and} \quad \dot{x}_n =u. 
\end{equation}
Letting $x:=(x_1,\ldots,x_n)^T$, system \rref{mult_int} can be compactly written as $\dot{x}= J_n x +e_n u$.

Inspired by \cite{Teel92}, our design of a $p$-bounded feedback for this system is based on a nested saturations feedback. We focus on the specific class of saturations that are linear around zero, and constant for large values of their argument.

\begin{defi}
\label{def:S(p)}
Given $p\in\mathbb N$, $\scal (p)$ is defined as the set of all odd functions $ \sigma $ of class $ C^{p}(\reels , \reels)$ such that there exist positive constants $\alpha$, $L$, $\sigma^{max}$ and $S$ satisfying, for all $r \in \reels$, $r \sigma (r) > 0$ when $r \neq 0$, $\sigma (r)= \alpha r $ for all $\abs{r} \leq L$, and $\abs{\sigma (r)} =\sigma^{max} $, when $\abs{r} \geq S$.
In the sequel, we associate with every \emph{$\sigma \in \scal(p)$ the $4$-tuple $(\sigma^{max},L,S,\alpha)$}.
\end{defi}
\begin{figure}[thpb]
      \centering
      \includegraphics[scale=0.27]{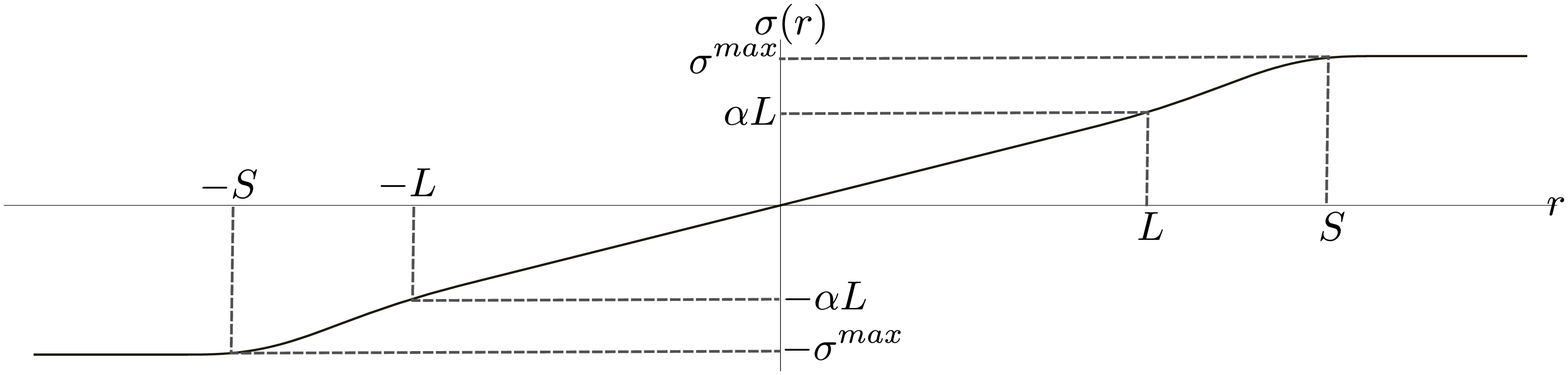}
      \caption{A typical example of a $\scal(p)$ saturation function with constants $(\sigma^{max},L,S,\alpha)$.}
      \label{fig:S(p)}
\end{figure}
The constants $\sigma^{max},L,S,\alpha$ will be widely used throughout the paper, see Fig. \ref{fig:S(p)} to fix ideas. Notice that it necessarily holds that $S \geq L$ and the equality may only hold when $p = 0$. We also stress that the successive derivatives up to order $p$ of an element of $\scal(p)$ are bounded. An example of such function is given in Section \ref{sec:simu} for $p=2$. 
The first result of this paper establishes that global stabilization on any chain of integrators by bounded feedback with constrained $p$ first derivatives can always be achieved by a particular choice of nested saturations. In other words, it solves the Problem in Case 1.
\begin{theo}
\label{Main_res}
Given $n\in\mathbb N_{\geq 1}$ and $p\in\mathbb N$, let $(R_j)_{0 \leq j \leq p}$ be a family of positive constants. For every set of saturation functions  $ \sigma_1 , \ldots , \sigma_n \in \scal(p)$, there exist vectors $ k_1 , \ldots , k_n $ in $\reels^n$, and positive constants  $ a_1 , \ldots , a_n $ such that the feedback law $\nu$ defined, for each $x\in\mathbb R^n$, as
\begin{align}
\label{nested_com_th}
\nu (x)= - a_n \sigma_n \Big( k_n^T x+a_{n-1} \sigma_{n-1}\big(k_{n-1}^T x +\ldots  + a_1  \sigma_1 (k_1^T x )\big) \ldots \Big)
\end{align}
is a $p$-bounded feedback law by $(R_j)_{0 \leq j \leq p}$ for system \rref{mult_int}, and the origin of the closed-loop system $\dot{x}= J_n x +e_n \nu(x)$ is GAS.
\end{theo}
The proof of this result is given in Section \ref{sec:proof_main_result} and the argument also provides an explicit choice of the gain vectors $ k_1, \ldots, k_n$ and constants $a_1,\ldots a_n $. 
\begin{remark}\label{rem:0}
Along the proof of Theorem \ref{Main_res} that the proposed construction allows to chose the magnitude of control signal independently of the magnitude of its $p$ first derivatives. More precisely, $a_n$ can be chosen to ensure that $\max \lbrace \abs{\nu (x)} : x \in \reels^n \rbrace= R_0$ and the gain vectors $ k_1, \ldots, k_n$ and constants $a_1,\ldots a_{n-1} $ can be taken in such a way that the $p$ first derivatives of the feedback are bounded by $(R_j)_{1 \leq j \leq p}$. 
\end{remark}
\begin{remark}
\label{rem:1}
In \cite{SSY}, a stabilizing feedback law was constructed using linear combinations of saturated functions. That feedback with saturation functions in $\scal(p)$ cannot be a $p$-bounded feedback for System \rref{mult_int} for $p\geq 1$. To see this, consider the double integrator, given by $\dot{x}_1=x_2, \: \dot{x}_2=u$. Any stabilizing feedback using a linear combination of saturation functions in $\scal (p)$ is given by $\nu (x_1,x_2) = - a  \sigma_1(b x_2) - c \sigma_2(d (x_2 + x_1))$, where the constants $a$, $b$, $c$, and $d$ are chosen to insure stability of the closed-loop system according to \cite{SSY}. Let $U(t)= \nu( x_1(t),x_2(t))$ for all $t \geq 0$. A straightforward computation yields, for $t\geq 0$, $\dot{U}(t) =- a b \sigma_{1}^{(1)}(a x_2 (t))  U(t) - c  d \sigma^{(1)}_{2}(d (x_2(t) + x_1(t))) (x_2(t) + U(t)).$
Consider now consider a solution with initial condition $x_2(0)=x_{20}$, and $x_1(0)=-x_{20}$ such that $\sigma_{1}^{(1)}(a x_{20})=0$. We then have $\dot{U}(0)  =  -cd \sigma^{(1)}_2(0) (x_{20} + U(0))$, whose norm is greater than $c_1(\vert x_{20}\vert -c_2)$ for some positive constants $c_1,c_2$. Thus $|\dot u(0)|$ grows unbounded as $|x_{20}|$ tends to infinity, which contradicts the definition of a $1$-bounded feedback.
\end{remark}

\subsection{Harmonic oscillators}

In Case 2, up to a linear change of coordinates, $A$ can be put in a block-diagonal form with skew-symmetric matrices on the diagonal. Addressing the stabilization problem under concern therefore amounts to only considering the following control system : $\dot{x}= Ax +bu$, where $x \in \reels^n$, $A \in \reels^{n, n}$ is skew-symmetric, $b \in \reels^n$ and the pair $(A,b)$ is controllable.

Unfortunately, the nested saturations feedback law given in \rref{nested_com_th} is a generic solution to the Problem for this class of systems only in the scalar case ($n=1$) or when when no rate constraint is imposed ($p=0$). To see why it may fail for $n\geq 2$, consider for instance the harmonic oscillator given by $\dot{x}_1= x_2$, $\dot{x}_2 = -x_1 + u $ (which we address in more details in Section \ref{sec:simu}) with a bounded stabilizing law given by $u =- \sigma (x_2)$ with $\sigma \in \scal (p)$ for some integer $p$. The time derivative of $u$ then satisfies $\abs{\dot{u}(t)} \geq \abs{\sigma^{(1)}(x_2(t))}(\abs{x_1(t)} - \abs{u(t)})$,
 which grows unbounded as $x_1$ goes unbounded and $\vert x_2\vert$ remains small (i.e. in the linear zone of $\sigma$). This prevents the feedback $- \sigma (x_2)$ to be a $1$-bounded feedback, hence a $p$-bounded feedback for all $p\geq 1$.

Our second result provides an alternative $p$-bounded feedback for the harmonic oscillator, thus solving the Problem in Case 2.

\begin{theo}
\label{theo2}
Given $n \in \entiers_{\geq 1}$ and $p\in\mathbb N$, let $(R_j)_{0 \leq j \leq p}$ be a family of positive constants, let $A \in \reels^{n, n}$ be a skew-symmetric matrix, and let $b \in \reels^n$ be such that the pair $(A,b)$ is controllable. Then, for any $\alpha \geq 1/2 $, there exists a positive constant $\beta$ such that the feedback law $\nu : \reels^n \to \reels$ defined as
\begin{equation}
\label{Yfeed}
\nu (x) := - \beta b^T x/ (1 + \norme{x}^2)^\alpha )
\end{equation} 
is a $p$-bounded feedback law by $(R_j)_{0 \leq j \leq p}$ for $\dot{x}= Ax +bu$, and the origin of the closed-loop system $\dot{x}=  A x + b \nu(x)$  is GAS.
\end{theo}

The proof of this theorem is given in Section \ref{subsec:th2}. 

\begin{remark}\label{rem:3}
Unlike for multiple integrators (see Remark \ref{rem:0}), the magnitude of the proposed feedback is not independently of the amplitude of its $p$ first derivatives.
\end{remark}
\section{Proof of main results}\label{sec: proof}

\subsection{Multiple integrators}\label{sec:proof_main_result}
\label{sec:proof_th}

We start by estimating upper bounds on composed saturation functions of the class $\scal(p)$. This estimate, presented in Lemma \ref{lem:techn}, relies on Fa\`a di Bruno's formula recalled below.

\begin{lem}[Fa\`a Di Bruno's formula, \cite{fdb}, p. 96]
\label{lem:fa_di}
Given  $k\in\mathbb N$, let $\phi\in C^{k}( \reels_{\geq 0} , \reels )$ and $\rho\in C^{k}( \reels , \reels )$. Then the $k$-th order derivative of the composite function $\rho \circ \phi$ is given by
\begin{equation}
\label{eq:faadibruno}
[ \rho \circ \phi ]^{(k)}(t) =  \sum\limits_{a=1}^k \rho^{(a)} (\phi(t)) B_{k,a}\Big(\phi^{(1)}(t), \ldots , \phi^{(k-a+1)}(t)\Big),
\end{equation}
where $B_{k,a}$ is the Bell polynomial given by
\begin{align}
\label{bell}
B_{k,a}\Big(\phi^{(1)}(t), \ldots ,  \phi^{(k-a+1)}(t)\Big)\hspace{-1mm}:=\hspace{-2mm}\sum\limits_{\delta \in \pcal_{k,a}} \hspace{-1mm}c_{\delta} \hspace{-1mm}\prod\limits_{l=1}^{k-a+1} \left( \phi^{(l)}(t) \right)^{\delta_l}
\end{align}
where $\pcal_{k,a}$ denotes the set of $(k-a+1)-$tuples $\delta :=(\delta_1 , \delta_2, \ldots , \delta_{k-a+1})$  of positive integers satisfying $\delta_1 + \delta_2 + \ldots +\delta_{k-a+1} = a$ and $\delta_1 +2 \delta_2 + \ldots + (k-a+1) \delta_{k-a+1} = k$, and $c_{\delta}:=k!/\left(\delta_1 ! \cdots \delta_{k-a+1}! (1!)^{\delta_1} \cdots ((k-a+1)!)^{\delta_{k-a+1}}\right)$.
\end{lem}

\begin{remark} 
\label{rem:2}
We stress that the Bell polynomial $B_{k,a}$ is of (homogeneous) degree $a$ w.r.t the $(k-a+1)$-dimensional vector representing the argument of $B_{k,a}$.
\end{remark}

The proof of Theorem \ref{Main_res} extensively relies on the following upper bound on composition of functions of $\scal (p)$, which exploits their constant value in their saturation region.

\begin{lem}
\label{lem:techn}
Given $k\in\mathbb N$, let $f$ and $g$ be functions of class $C^k(\reels_{\geq 0} , \reels)$, $\sigma$ be a saturation function in $\scal (k)$ with constants ($\sigma^{max}, L, S ,\alpha$), and $E$ and $F$ be subsets of $\reels_{\geq 0}$ such that $E \subseteq F$. Assume that
\begin{eqnarray}
\label{eq:lem:tech:absf}
& \: \abs{f(t)} > S ,& \:  \quad  \forall t \in F \backslash E,
\end{eqnarray} 
and that there exist positive constants $M, Q_1, \ldots , Q_k$ such that
\begin{equation}
\label{eq:lem:tech:fkQ}
\abs{g^{(k)}(t)} \leq M, \:  \forall t \in F \quad \text{and} \quad \abs{f^{(a)}(t)} \leq Q_{a}, \: \forall t \in E, \: \forall  a \in \llbracket 1 , k \rrbracket.
\end{equation}

Set $\overline{\sigma}_a := \max \lbrace |\sigma^{(a)}(s)| : s \in \reels \rbrace$ for each $a\in \llbracket 1 , k \rrbracket$. Then the $k$th-order derivative of the function  $h : \reels_{\geq 0} \rightarrow \reels$ defined by $h(\cdot) := g(\cdot) + \sigma ( f(\cdot))$, satisfies
\begin{equation}
\label{lem:eq:est}
\abs{h^{(k)}(t)} \leq M + \sum\limits_{a=1}^{k}  \overline{\sigma}_a B_{k,a}(Q_1, \ldots , Q_{k -a+1})  ,  \quad \forall  t   \in F . 
\end{equation}
\end{lem}
\begin{proof}[Proof of Lemma \ref{lem:techn}]
Using Lemma \ref{lem:fa_di}, a straightforward computation yields
\begin{equation*}
h^{(k)}(t) = g^{(k)}(t) + \sum_{a=1}^{k} \sigma^{(a)} (f(t))  B_{k,a}\Big(f^{(1)}(t), \ldots , f^{(k-a+1)}(t)\Big), \quad \forall t \geq 0,
\end{equation*} 
where the polynomials $B_{k,a}$ are defined in \rref{bell}. Since $\sigma \in \scal(k)$, \rref{eq:lem:tech:absf} ensures that the set $F\setminus E$ is contained in the saturation zone of $\sigma$. It follows that
\begin{equation}
\label{proof:lem:tech:int}
[ \sigma \circ f ]^{(k)}(t)=0,\quad \forall t\in F\setminus E.
\end{equation}
Furthermore, from \rref{eq:lem:tech:fkQ} and \rref{bell} it holds that, for all $t \in E$,
\begin{equation}
\abs{B_{k,a}\left( f^{(1)}(t), \ldots , f^{(k -a+1)}(t)\right)}
 \leq \sum_{\delta \in \pcal_{k,a}} c_{\delta} \prod\limits_{l=1}^{k-a+1} Q_l^{\delta_l} =B_{k,a}(Q_1, \ldots ,Q_{k-a+1}).
\end{equation}
One has $[ \sigma \circ f ]^{(k)}(t) \leq  \sum_{a=1}^{k} \overline{\sigma}_a B_{k,a}(Q_1, \ldots ,Q_{k -a+1})$, for all $ t \in E$, from the definition of $\overline{\sigma}_a $ and \rref{eq:faadibruno}. In view of \rref{proof:lem:tech:int}, the above estimate is valid on the whole set $F$. Thanks to \rref{eq:lem:tech:fkQ}, a straightforward computation leads to the estimate \rref{lem:eq:est}.
\end{proof}


We now turn to the proof of Theorem \ref{Main_res}. We explicitly construct the vectors $k_1, \ldots , k_n $ and the constants $a_1 , \ldots , a_n$ guaranteeing global asymptotic stability with a bounded feedback law whose successive derivatives remain below prescribed bounds at all times. Given $p\in\mathbb N$ and $n\in\mathbb N_{\geq 1}$, let $ \sigma_i $ be saturation functions in $\scal (p)$ with constants $( \sigma_i^{max} , L_{\sigma_i} , S_{\sigma_i}, \alpha_{\sigma_i} )$ for each $i \in \llbracket 1 , n \rrbracket$, and let $(R_j)_{0 \leq j \leq p}$ be the family of prescribed positive bounds on the amplitude and the successive time derivatives of the control signal. We first construct a $p$ bounded feedback by $(R_j)_{0 \leq j \leq p}$ for System \rref{mult_int}. Then we show that the origin of the closed loop system \rref{mult_int} with this feedback law is GAS.

Let $(\mu_i^{max})_{1 \leq i \leq n-1}$ and $(L_{\mu_i})_{1 \leq i \leq n-1}$ be two families of positive constants such that $\mu_{n-1}^{max} <  1/2$, $L_{\mu_{n-1}}  =    \mu_{n-1}^{max} L_{\sigma_{n-1}} \alpha_{\sigma_{n-1}}/  \sigma_{n-1}^{max}  $, and, for each $i \in \llbracket 1, n-2 \rrbracket$, $\mu_i^{max} <   L_{\mu_{i+1}} / 2 $, $L_{\mu_{i}}  =    \mu_{i}^{max} L_{\sigma_{i}} \alpha_{\sigma_{i}} / \sigma_{i}^{max}  $. For each  $i \in \llbracket 1, n-1 \rrbracket$, we define saturation function $\mu_i \in \scal (p)$ with constants ($\mu_i^{max}$, $ L_{\mu_i} $, $S_{\mu_i}$, $1 $), where $ S_{\mu_i} =   S_{\sigma_i}  L_{\mu_i}/ L_{\sigma_i}  $, as follows $\mu_i (s) :=  \mu_i^{max}   \sigma_i ( s  L_{\sigma_i} /L_{\mu_i}   )/  \sigma_i^{max}$, for all $s \in \reels$. For $\lambda \geq 1$, to be chosen later, we define the saturation function $\mu_n$ as $\mu_n (s)  := R_0 \sigma_n (s L_{ \sigma_n} / \lambda   ) / \sigma_n^{max}$, for all $s \in \reels$, with $\mu_n^{max}=R_0$,  $L_{\mu_n}=\lambda $, $S_{\mu_n}= S_{\sigma_n}  \lambda / L_{\sigma_n}$, and $\alpha_{\mu_n}:=\alpha_{\tilde{\mu}} / \lambda$ with $\alpha_{\tilde{\mu}}  := R_0 L_{\sigma_n} \alpha_{\sigma_n} / \sigma_n^{max} $.

We next make a linear change of coordinates $y = H x $, with $H\in \reels^{n,n}$, that puts System \rref{mult_int} into the form
\begin{equation}
\label{S1}
\dot{y}_i  =  \alpha_{\tilde{\mu}} / \lambda  \sum_{l=i+1}^n y_l + u, \quad \forall i \in \llbracket 1 , n \rrbracket.
\end{equation}
The relations $ y_{n-i}=\sum_{k=0}^i \binom{i}{k}  \left( \alpha_{\tilde{\mu}} / \lambda  \right)^k  x_{n-k}$, for $i \in \llbracket 0, n-1 \rrbracket$, enable us to determine $H$. Since $H$ is triangular with non zero elements on the main diagonal, it is invertible. We define a nested saturations feedback law $\Upsilon : \reels^n \rightarrow \reels $ as
\begin{equation}
\label{fe:prop:S_1}
\Upsilon(y)=- \mu_{n} (y_n + \mu_{n-1}(y_{n-1}+ \ldots + \mu_1(y_1))\ldots ). 
\end{equation}
Note that, in the original $x$-coordinates, this feedback law reads $\Upsilon(y)=\Upsilon(Hx)$ therefore the bounds of the successive time derivatives of $\Upsilon(y)$ coincide with those of $\Upsilon (Hx)$. The global stabilization of \rref{S1} with a $p$-bounded feedback law by $(R_j)_{0\leq j\leq p}$ is thus equivalent to that of the original system \rref{mult_int}. So, from now on, we will rely on this expression. Let $y(\cdot)$ be a trajectory of the system
\begin{equation}
\label{S1cc}
\dot{y}_i  =  \alpha_{\tilde{\mu}} / \lambda  \sum_{l=i+1}^n y_l + \Upsilon(y), \quad \forall i \in \llbracket 1 , n \rrbracket ,
\end{equation} which is the closed-loop system \rref{S1} with the feedback defined in \rref{fe:prop:S_1}. For each $i\in  \llbracket 1 , n\rrbracket$, let $z_i : \reels_{\geq 0} \rightarrow \reels$ be the time function defined recursively as $z_i(\cdot ): = y_i(\cdot) + \mu_{i-1}(z_{i-1}(\cdot))$ , with $\mu_0(\cdot)= 0$. With the above functions, the closed-loop system \rref{S1cc} can be rewritten as
\begin{equation}
\begin{cases}
\label{S1c2}
\dot{y}_i  = \alpha_{\tilde{\mu}} / \lambda z_n - \mu_{n} (z_n) + \alpha_{\tilde{\mu}} / \lambda (  \sum\limits_{l=i+1}^{n-1}\gamma_l(z_l) - \mu_i(z_i)), \: \forall i \in \llbracket 1 , n-1 \rrbracket ,\\
\dot{y}_n  =  -  \mu_{n} (z_n),
\end{cases}
\end{equation} where $\gamma_l(z_l)= z_l - \mu_{l} (z_l)$ for all $l \in \llbracket 1 , n-1 \rrbracket$.
For each $i\in  \llbracket 1 , n\rrbracket$, we also let $E_i  :=  \left\lbrace y \in \reels^{n} : \:  \abs{y_v } \leq S_{\mu_v} + \mu_{v-1}^{max}, \forall v \in \llbracket i, n \rrbracket  \right\rbrace $, and $I_i :=   \lbrace t \in \reels_{\geq 0} : \: y(t) \in E_i \rbrace $. Note that from the definitions of $I_i$ and $E_i$, we have $I_{1} \subseteq I_{2} \subseteq \ldots \subseteq I_n$, and a straightforward computation yields
\begin{align}
\abs{z_{i}(t)} &> S_{\mu_{i}} , \quad \forall t \in I_{i+1} \backslash I_{i}, \:   \forall i \in \llbracket 1 , n-1 \rrbracket , \label{eq:s_i_I_i_prive_I_i+1}\\
\abs{z_n(t)} &> S_{\mu_{n}} , \quad \forall t \in \reels_{\geq 0} \backslash I_{n}, \label{eq:s_n_rn}
\end{align}which allows to determine when saturation occurs. We define $b_{\mu_i}  :=  \max \lbrace \abs{r - \mu_i(r)} : \:\abs{r} \leq S_{\mu_i} + 2 \mu_{i-1}^{max}  \rbrace $ and $\overline{\mu}_{i,j}  :=  \max  \lbrace \abs{\mu_i^{(j)}(r)}: \:  r \in  \reels  \rbrace$ for each $i \in \llbracket 1 , n-1 \rrbracket $ and $j \in \llbracket 1 ,p \rrbracket $. Note that these quantities are well defined since the functions $\mu_i$ are all in $\scal (p)$. We can now establish that  
\begin{align}\label{eq:si-sig}
\abs{z_i(t) - \mu_i(z_i(t))} &\leq b_{\mu_i} , \quad \forall t \in I_i, \quad \forall i \in \llbracket 1 , n-1 \rrbracket.
\end{align}Set $\underline{b}_{\sigma_n} := \min \lbrace \sigma_n(r)/r : \: 0 < \abs{r} \leq S_{\sigma_n}  + 2 \mu_{n-1}^{max} L_{\sigma_n}  \rbrace$, $\overline{b}_{\sigma_n} := \max \lbrace \sigma_n(r)/r : \: \forall r > 0  \rbrace$, and $\Delta=(\overline{b}_{\sigma_n} - \underline{b}_{\sigma_n})(L_{\sigma_n} R_0 / \sigma_n^{max}) $. It then can been seen that
\begin{equation}\label{eq:sn-sig}
\abs{\alpha_{\mu_n} z_n(t)-\mu_n (z_n(t))} \leq \Delta ( S_{\sigma_n} / L_{\sigma_n} + 2 \mu_{n-1}^{max} / \lambda ), \quad \forall t \in I_n.
\end{equation}

The following statement provides explicit bounds on the successive derivatives of each functions $y_i(t)$, $z_i(t)$ for each $i  \in \llbracket 1, n \rrbracket$ and the control input given by $U(\cdot):=\Upsilon(y(\cdot))$.
\begin{assert}
\label{prop:2}
With the notation introduced previously and the Bell polynomials introduced in \eqref{bell}, every trajectory of the closed-loop system \rref{S1cc} satisfies, for each $i\in\llbracket 1,n\rrbracket$ and each $j \in \llbracket 1, p \rrbracket $,
\begin{align}
(P_1(i,j)) : \: & \: \abs{y_i^{(j)}(t)} \leq Y_{i,j} ,\: \forall t \in I_i \,; \quad (P_2(i,j)) :  \: \abs{z_i^{(j)}(t)} \leq  Z_{i,j} , \: \forall t \in I_i ,   \label{y_j_est}\\ 
(P_3(j)) :\quad & \:\sup \lbrace \abs{U^{(j)}(t)} : t\geq 0 \rbrace   \leq   \sum\limits_{q=1}^j  G_{q,j} \overline{\mu}_{n,q} \,;\label{sup_u_j_est}
\end{align}
where $\overline{\mu}_{n,q}= \max \lbrace \abs{\mu_n^{(q)}(r)}: \:  r \in  \reels  \rbrace $, $Y_{i,j}$, $Z_{i,j}$, and $G_{q,j}$ are independent of initial conditions and are obtained recursively as follows: for $j=1$, $Y_{i,1} :=  \Delta ( S_{\sigma_n} / L_{\sigma_n} + 2 \mu_{n-1}^{max} / \lambda ) +  \alpha_{\tilde{\mu}} / \lambda  ( \sum_{l=i+1}^{n-1} b_{\mu_l} +  \mu_i^{max})$ for $i \in \llbracket 1, n-1 \rrbracket $, $Y_{n,1} := \mu_n^{max}$, $Z_{1,1}  :=  Y_{1,1}$, $Z_{i,1}  :=  Y_{i,1} + \overline{\mu}_{i-1,1} Z_{i-1,1}$ for $i \in \llbracket 2, n \rrbracket$, and $G_{1,1}  := Z_{n,1}$. For each $j\in\llbracket 2, p\rrbracket $, the inductive relations are given by $Y_{i,j} :=  \alpha_{\tilde{\mu}} / \lambda \sum_{b=i+1}^{n} Y_{b,j-1} + \sum_{q=1}^{j-1} G_{q,j-1} \overline{\mu}_{n,q}$ for $i \in \llbracket 1, n \rrbracket $, $Z_{1,j}  :=  Y_{1,j}$, $Z_{i,j}  :=  Y_{i,j} + \sum_{a=1}^j \overline{\mu}_{i-1,a} B_{j,a}(Z_{i-1,1}, \ldots  , Z_{i-1,j-1+a})$, for  $i \in \llbracket 2, n \rrbracket $, and $G_{q,j} := B_{j,q}(Z_{n,1}, \ldots , Z_{n,j-q+1})$, for $q \in \llbracket 1, j \rrbracket $.

\end{assert}
\begin{proof}[Proof of Assertion \ref{prop:2}]

The right-hand side of System \rref{S1cc} is globally Lipschitz and of class $C^p( \reels^n , \reels^n )$, and therefore it is forward complete with trajectories of class $C^{p+1}( \reels_{\geq 0} , \reels^n )$. We establish the result by induction on $j$. We start by $j=1$. We begin to prove that $P_1(i,1)$  holds for all $i \in \llbracket 1, n \rrbracket $. Let $i \in \llbracket 1, n-1 \rrbracket $. From \rref{S1c2}, \rref{eq:si-sig}, and \rref{eq:sn-sig} a straightforward computation leads to $\abs{ \dot{y}_i(t) } \leq   \Delta ( S_{\sigma_n} / L_{\sigma_n} + 2 \mu_{n-1}^{max} / \lambda ) +  \alpha_{\tilde{\mu}} / \lambda  ( \sum_{l=i+1}^{n-1} b_{\mu_l} +  \mu_i^{max})$, for all $t \in I_{i+1}$. Since $I_{i} \subseteq I_{i+1}$, the above estimate is still true on $I_i$. Moreover, from \rref{S1c2} it holds that $\abs{ \dot{y}_n(t) } \leq \mu_n^{max}$ at all positive time. Thus, $P_1(i,1)$ has been proven for each $i \in \llbracket 1, n \rrbracket $. 

We now prove by induction on $i$ the statement $P_2(i,1)$. Since $z_1(\cdot )=y_1(\cdot )$, the case $i=1$ is done. Assume that, for $i \in \llbracket 1, n-1 \rrbracket $, $P_2(i, 1)$ holds. From Lemma \ref{lem:techn} (with $k=1$, $f(\cdot)=z_i(\cdot)$, $g(\cdot)=y_{i+1}(\cdot)$, $h(\cdot)=z_{i+1}(\cdot)$, $\sigma(\cdot) = \mu_i(\cdot)$, $Q_1= Z_{i,1}$, $M = Y_{i+1,1}$, $\overline{\sigma}_1 = \overline{\mu}_{i,1}$, $E=I_i$, $F=I_{i+1}$, and \rref{eq:s_i_I_i_prive_I_i+1}), we can establish that $P_2(i+1,1)$ holds. Thus $P_2(i,1)$ holds for all $i \in \llbracket 1, n \rrbracket $.

Notice that the applied control input reads $U(\cdot)=-\mu_n(z_{n}(\cdot))$. We then can establish $P_3(1)$ from Lemma \ref{lem:techn} (with $k=1$, $f(\cdot)=z_n(\cdot)$, $g(\cdot)\equiv 0$, $h(\cdot)=U(\cdot)$, $\sigma(\cdot) = \mu_n(\cdot)$, $Q_1= Z_{n,1}$, $M = 0$, $\overline{\sigma}_1 = \overline{\mu}_{n,1}$, $E=I_n$, $F=\reels_{\geq 0}$ and \rref{eq:s_n_rn}). This ends the case $j=1$.

Now, assume that for a given $j \in \llbracket 1, p-1 \rrbracket$, statements $P_1(i,j_2)$, $P_2(i,j_2)$ and $P_3(j_2)$ hold for all $j_2 \leq j$ and all $i \in \llbracket 1, n \rrbracket $. Let $i \in \llbracket 1, n \rrbracket $. From \rref{S1}, a straightforward computation yields $\abs{ y_i^{(j+1)}(t) }  \leq  \alpha_{\tilde{\mu}} / \lambda  \sum_{l=i+1}^{n} \abs{ y_l^{(j)}(t) } + \abs{u^{(j)}(t)}$, for all $ t \geq 0.$
Using $P_3(j)$, $P_1(i+1,j), \ldots , P_1(n,j) $, we obtain that $\abs{ y_i^{(j+1)}(t) }  \leq  \alpha_{\tilde{\mu}} / \lambda  \sum_{l=i+1}^{n} Y_{l,j} + \sum_{q=1}^j  G_{q,j}\overline{\mu}_{n,q}$ , for all $l t \geq I_i.$ Thus the statement $P_1(j+1,i)$ is proven for all $i \in \llbracket 1, n \rrbracket $.

We now prove by induction on $i$ the statement $P_2(i,j+1)$. As before, since $z_1(\cdot)=y_1(\cdot)$, the case for $i=1$ is done. Assume that for a given $i \in \llbracket 1, n-1 \rrbracket$, the statement $P_2(i, j+1)$ holds. From Lemma \ref{lem:techn} (with $k=j+1$, $f(\cdot)=z_i(\cdot)$, $g(\cdot)=y_{i+1}(\cdot)$, $h(\cdot)=z_{i+1}(\cdot)$, $\sigma(\cdot) = \mu_i(\cdot)$, $Q_{k_1}= Z_{i,k_1}$, $M = Y_{i+1,j+1}$, $\overline{\sigma}_a = \overline{\mu}_{i,a}$, $E=I_i$, $F=I_{i+1}$, and \rref{eq:s_i_I_i_prive_I_i+1}), we can establish that $P_2(i+1,j+1)$ holds. $P_2(i,j+1)$ is thus satisfied for all $i \in \llbracket 1, n \rrbracket $.

Finally, we can establish $P_3(j+1)$ from Lemma \ref{lem:techn} (with $k=j+1$, $f(\cdot)=z_n(\cdot)$, $g(\cdot)\equiv 0$, $h(\cdot)=U(\cdot)$, $\sigma(\cdot) = \mu_n(\cdot)$, $Q_{k1}= Z_{n,k_1}$, $M = 0$, $\overline{\sigma}_a = \overline{\mu}_{n,a}$, $E=I_n$, $F=\reels_{\geq 0}$ and \rref{eq:s_n_rn}). This ends the proof of Assertion \ref{prop:2}.
\end{proof}

First notice that $\overline{\mu}_{n,q} = \tilde{\mu}_{n,q} / \lambda^q$ with $\tilde{\mu}_{n,q}:=  R_0  \overline{\sigma}_{n,q} (L_{\sigma_n})^q / \sigma_n^{max} $  and $\overline{\sigma}_{n,q} := \max \lbrace \abs{\sigma^{(j)}(r)}: \:  r \in  \reels  \rbrace$ for $q \in \llbracket 1 , p \rrbracket $. In consequence, using Assertion \ref{prop:2}, it can be seen that, for each  $j\in \llbracket  1,p \rrbracket $, $\sum_{q=1}^j   G_{q,j} \tilde{\mu}_{n,q} / \lambda^q = P( 1/ \lambda )/\lambda$ where $P$ is a polynomial with positive coefficients. This sum is thus decreasing in $\lambda$. Hence, we can pick $\lambda\geq 1$ in such a way that $\sum_{q=1}^j   G_{q,j} \tilde{\mu}_{n,q}/ \lambda^q  \leq \underline{R}$, for each $ j \in \llbracket 1, p \rrbracket$ with $\underline{R}:= \min \lbrace R_1 , \ldots , R_p \rbrace $. It follows that, for each  $j\in\llbracket1 , p\rrbracket$, $\sup \lbrace \abs{U^{(j)}(t)} : t\geq 0 \rbrace \leq   \underline{R} \leq R_j$. By recalling that the feedback $\Upsilon$ is bounded by $R_0$, we conclude that it is $p$-bounded feedback law by $(R_j)_{0 \leq j \leq p}$ for System \rref{S1}.

It remains to prove that the feedback law \rref{fe:prop:S_1}, where now all the coefficient have been chosen, stabilizes System \rref{S1}. Actually the proof is almost the same as the one given in \cite{Teel92}, except that we allow the first level of saturation $\mu_n$ to have a slope different from $1$.

We prove that after a finite time any trajectory of the closed-loop system \rref{S1cc} enters a region in which the feedback \rref{fe:prop:S_1} becomes simply linear. To that end, we consider the Lyapunov function candidate $V_n(y_n):= \frac{1}{2} y_n^2$. Its derivative along the trajectories of \rref{S1cc} reads $\dot{V}_n= -y_n  \mu_n (y_n + \mu_{n-1}(z_{n-1}))$. From the choice of $\lambda$ and $\mu_{n-1}^{max}$, we obtain the following implication, with $\theta:= \inf\limits_{ r \in  [L_{\mu_n}/2 - \mu_{n-1}^{max} , S_{\mu_n} ]} \left\lbrace \mu_n (r) \right\rbrace $, 
\begin{equation}
\label{lyap}
\abs{y_n} \geq L_{\mu_{n}}/2\quad\Rightarrow\quad \dot{V}_n \leq  - \theta L_{\mu_{n}}/2 .
\end{equation}

We next show that there exists a time $T_1 \geq 0$ such that $ \abs{y_n(t)} \leq L_{\mu_{n}}/2 $, for all $t \geq T_1$.  To prove that, we consider the following alternative: either for every $t \geq 0 $, $\abs{y_n(t)} \leq L_{\mu_n} /2 $ and we are done, or there exist $T_0 \geq 0$ such that $\abs{y_n(T_0)} > L_{\mu_n} /2 $. In that case there exists $ \tilde{T}_0 \geq T_
0 $ such that $y_n(\tilde{T}_0)=L_{\mu_n} /2$ (otherwise thanks to \rref{lyap}, $ V_n(t) \rightarrow - \infty $ as $t \rightarrow \infty $ which is impossible). Due to \rref{lyap}, we have  $\abs{y_n(t)} < L_{\mu_n} /2$ in a right open neighbourhood of $\tilde{T}_0$. Suppose that there exists a positive time $\tilde{T}_1 > \tilde{T}_0$ such that $\abs{y_n(\tilde{T}_1)} \geq L_{\mu_{n}}/2 $. Then by continuity, there must exists $\tilde{T}_2 \in ( \tilde{T}_0 , \tilde{T}_1]$ such that 
$\abs{y_n(\tilde{T}_2)}= L_{\mu_n}/2$, and $\abs{y_n(t)} < L_{\mu_n}/2$ for all $t \in ( \tilde{T}_0 , \tilde{T}_2)$.
However, it then follows from \rref{lyap} that for a left open neighbourhood of $\tilde{T}_2$ we have $\abs{y_n(t)}> \abs{y_n(\tilde{T}_2)}=L_{\mu_n}/2$. This is a contradiction with the fact that on a right open neighbourhood of $\tilde{T}_0$ we have $\abs{y_n(t)} < L_{\mu_n} /2$. 
Therefore, for every $\tilde{T}_1 > \tilde{T}_0$, one has $\abs{y_n(\tilde{T}_1)} < L_{\mu_n} /2$ and the claim is proved.

In consequence we have that $\abs{y_n(t) + \mu_{n-1}(z_{n-1} (t))} \leq  L_{\mu_n}$ for all $t \geq T_1$.
Therefore $\mu_n$ operates in its linear region after time $T_1$. Similarly, we now consider $V_{n-1}(y_{n-1}):= \frac{1}{2} y_{n-1}^2$, whose derivative along the trajectories of \rref{S1cc} satisfies $\dot{V}_n= -  \alpha_{\mu_n}  y_{n-1} \mu_{n-1}(y_{n-1}+ \mu_{n-2}(y_{n-2}+ \ldots))$, for all $ t \geq T_1$.
Reasoning as before, there exists a time $T_2 >0$ such that $\abs{y_{n-1}(t)} \leq L_{\mu_{n-1}}/2$ and $\mu_{n-1}$ operates in its linear region for all $t \geq T_2$.

By repeating this procedure, we construct a time $T_n$ such that for every $t\geq T_n$ the whole feedback law becomes linear, that is $\Upsilon(y(t)) = -\alpha_{\mu_n} ( y_n(t) + \ldots + y_1(t) )$, when $t \geq T_n.$
Thus, after time $T_n$, the system \rref{S1cc} becomes linear and its local exponential stability follows readily. Thus the origin of the closed-loop system \rref{S1cc} is globally asymptotically stable.

With the linear change $y=Hx$ proposed in \rref{S1cc}, the closed-loop system \rref{S1cc} can be put into the form of the original system \rref{mult_int} in closed loop with $u=\Upsilon(Hx)$. Thus, the sought feedback law $\nu$ of Theorem \ref{Main_res} is obtained by $\nu(x)=\Upsilon(Hx)$. This leads to the following choices of parameters: for each $i \in  \llbracket 1, n-1$, $a_i  = L_{\sigma_{i+1}} \mu_i^{max} / (L_{\mu_{i+1}}     \sigma_i^{max})$,  $k_{n-i}^T x =  L_{\sigma_{n-i}}/L_{\mu_{n-i}}  \sum_{k=0}^i \binom{i}{k}( \alpha_{\tilde{\mu}}/ L_{\mu_{n}} )^k  x_{n-k}$,
$a_n = R_0 / \sigma_n^{max}$, and $k_n^T x =   x_n L_{\sigma_n}/ L_{\mu_{n}} $.

\subsection{Skew-symmetric systems with scalar input}
\label{subsec:th2}
Given $n\in\mathbb N_{\geq 1}$, let $A \in \reels^{n, n}$ be a skew-symmetric matrix and $b \in \reels^n$ such that $(A,b)$ is controllable.  Let $p\in\mathbb N$, $\alpha \geq 1/2 $ and $(R_j)_{0 \leq j \leq p}$ be a family of positive constants. The proof of Theorem \ref{theo2} is divided into two steps. We first prove that for any $\beta > 0$ and $\alpha \geq 1/2$, the origin of $\dot{x}=Ax+bu$ with the feedback law \rref{Yfeed} is GAS. We then show that for any $\alpha \geq 1/2$ there exists a positive constant $\beta$ such the feedback law \rref{Yfeed} is a $p$-bounded feedback law by $(R_j)_{0 \leq j \leq p}$ for system $\dot{x}=Ax+bu$.

Let $\beta$ be a positive constant and $\alpha \geq 1/2$. We define $A_{\beta}:=  A- \beta b b^{T}$. The matrix $A_\beta$ is then Hurwitz. To see this, observe that the Lyapunov equation $A_\beta^T+A_\beta =-2\beta b b^T$ holds and that the pair $(A_\beta,b)$ is controllable. Thus there exists a real symmetric positive definite matrix $P\in\reels^{n,n}$ such that $P A_{\beta} + A_{\beta}^T P = - \identity_n$.

Let $V(x):= x^T Px + K ( (1 + \norme{x}^2 )^{\alpha +1} -1  )$ be a candidate Lyapunov function with  $K:=( \norme{ Pb }  \beta )^2 / (\alpha +1) $. The derivative of $V$ along trajectories of $\dot{x}=Ax+b\nu(x)$ is then given by
\begin{align*}
\dot{V}(x) =  - \norme{x}^2 + 2 \beta ( 1 -1/(1 + \norme{x}^2)^\alpha  ) x^T P b b^T x  - 2 K( \alpha +1 ) ( b^T x )^2.
\end{align*} 
Using that $\abs{x^T P b} \leq \norme{x} \norme{P b}$ and 
$$
2 \beta ( 1 -1/(1 + \norme{x}^2 )^\alpha )  \abs{x^T P b b^T x }\leq
 \norme{x}^2 / 2  + 2  \beta^2 \norme{P b}^2 (b^T x)^2 , 
$$
we get that $\dot{V}(x) \leq -\norme{x}^2/2$. Therefore the origin of $\dot{x}=Ax+b\nu(x)$ is GAS.

We next give and prove two assertions which give explicit formula of the successive derivatives of the trajectories and of the control signal. 
\begin{assert}
\label{assert1}
Given any $\beta >0$ and $\alpha \geq 1/2$, each trajectory $x: \reels_{\geq 0} \to \reels^n$ of  $\dot{x}=Ax+b\nu(x)$ is $C^{\infty}$ and satisfies, for any $k \in \entiers_{\geq 1}$, with $U(\cdot) := \nu(x(\cdot))$, 
\begin{equation}
\label{eq:assert1}
x^{(k)}(t) = A^k x(t) + \sum\limits_{j=0}^{k-1} A^j b U^{(k-(j+1))}(t), \quad \forall t \geq 0.
\end{equation}
\end{assert}
\begin{proof}[Proof of Assertion \ref{assert1}]
The right hand side of $\dot{x}=Ax+b\nu(x)$ is globally and  Lipschitz $C^{\infty}$, and therefore this system is forward complete with trajectories are of class $C^{\infty}$ $(\reels_{\geq 0}$, $\reels^n)$. In particular, the successive derivatives of $U(\cdot)$ and $x(\cdot)$ are well defined. Equation \rref{eq:assert1} then follows by a trivial induction argument using differentiation of $\dot{x}=Ax+b\nu(x)$ at any order.
\end{proof}

\begin{assert}
\label{assert2}
Given $\beta >0$ and $\alpha \geq 1/2$, let $x: \reels_{\geq 0} \to \reels^n$ be any trajectory of $\dot{x}=Ax+b\nu(x)$ and let $G(\cdot):= 1 + \norme{x(\cdot)}^2 $. Then, for any $k \in \entiers_{\geq 1}$ and all $t\in\reels_{\geq 0}$, it holds that 
\begin{equation}
\label{eq:dg}
G^{(k)}(t)
 = \sum\limits_{m=0}^{k} \binom{k}{m} \sum\limits_{i=1}^n x^{(m)}_i(t) x^{(k - m)}_i(t), 
\end{equation} 
and, with $d_a := (-1)^a \prod\limits_{i=0}^a ( \alpha + i)$ and $B_{l,a}$ introduced in \rref{bell},  
\begin{eqnarray}
\label{assrt2:eq:u}
U^{(k)}(t) & = & - \beta \Big( \frac{b^{T} x^{(k)}(t)}{(1 + \norme{x(t)}^2)^\alpha } \nonumber \\
 & + & \sum\limits_{l=1}^k \binom{k}{l} \sum\limits_{a=1}^l d_a \frac{ B_{l,a}(G^{(1)}(t) , \ldots , G^{(l-a+1)}(t) ) b^T x^{(k-l)}(t) }{(1 + \norme{x(t)}^2)^{\alpha + a}} \Big).
\end{eqnarray} 

\end{assert}

\begin{proof}[Proof of Assertion \ref{assert2}]
Expression \rref{eq:dg} is readily obtained from the general Leibniz rule. In order to establish \rref{assrt2:eq:u}, let $f: \reels_{>0} \to \reels_{>0}$ be defined as $f(z):=z^{-\alpha}$. The feedback law can then be rewritten as $U(t)= - \beta b^T x(t) [f \circ G](t)$. Using the general Leibniz rule we get that, for any $k \in \entiers_{\geq 1}$ and  any $t\in\reels_{\geq 0}$, $U^{(k)}(t)= - \beta b^T x^{(k)}(t) [ f \circ G](t)  + \sum_{l=1}^k \binom{k}{l} [ f \circ G]^{(l)} ( t) b^T x^{(k-l)}(t)$. Thanks to Fa\`a Di Bruno's formula (Lemma \ref{lem:fa_di}), we obtain that, for each $l \in \llbracket 1 , k \rrbracket $ and each $t\in\reels_{\geq 0}$, $[ f \circ G]^{(l)}(t)=\sum_{a=1}^l f^{(a)}( G(t) ) B_{l,a} ( G^{(1)}(t) , \ldots , G^{(l-a+1)}(t) ).$
Since $f^{(a)}(z) =(-1)^a \prod_{i=0}^a ( \alpha + i) z^{-(a+\alpha) }$ for all $a \in \llbracket 1 , k \rrbracket $,  Assertion \ref{assert2} is proven.
\end{proof}

We now proceed with Step 2. Set $\underline{R}:= \min \left\lbrace R_0 , \ldots , R_p \right\rbrace$. We now prove by induction on $j \in \llbracket 0, p \rrbracket$ that there exist $\beta_j > 0$ such that, for any $\beta \leq \beta_j $, each trajectory of $\dot{x}=Ax+b\nu(x)$ satisfies $\sup \lbrace \abs{U^{(j_1)}(t)} : t\geq 0\rbrace \leq \underline{R}$ for $j_1 \in \llbracket 1 , j \rrbracket$.

Note that this in turn ensures that the feedback law \rref{Yfeed} is a $j$-bounded feedback law by $(R_{j_1})_{0 \leq j_1 \leq j}$.

We start by $j=0$. Since $\sup \lbrace \abs{ \nu (x)} : x \in \reels^n \rbrace \leq \beta $, the base case follows by choosing $\beta_0 = \underline{R}$. Now assume that, for a given $j \in \llbracket 0, p-1 \rrbracket$, there exists $\beta_j > 0$ such that for any $\beta \leq \beta_j $ the feedback law \rref{Yfeed} is a $j$-bounded feedback law by $(R_{j_1})_{0 \leq j_1 \leq j}$. Using Assertion \ref{assert1}, we get that for each $\beta \leq \beta_j $ there exists for each $i \in \llbracket 1 , n \rrbracket $ and $k \in \llbracket 1 , j+1 \rrbracket$ a multivariate polynomial $P_{i,k} : \reels^n \to \reels $ of degree $1$ (which not depend on $\beta$)  such that
\begin{eqnarray}
\label{eq:pr:th2:absx}
\abs{x^{(k)}_i(t)} \leq P_{i,k}(\abs{x_1(t)}, \ldots , \abs{x_2(t)}) ,\quad \forall t \geq 0.
\end{eqnarray}
From \rref{eq:dg} and \rref{eq:pr:th2:absx} it follows that, for each $k \in \llbracket 1 , j+1 \rrbracket$, there exists a multivariate polynomial $\underline{P}_{k} : \reels^n \to \reels $ of degree $2$, which not depend on $\beta$, such that
\begin{equation}
\label{eq:pr:th2:absG}
\abs{G^{(k)}(t)} \leq \underline{P}_{k}(\abs{x_1(t)}, \ldots , \abs{x_2(t)}), \quad \forall t \geq 0.
\end{equation}
In view of Remark \ref{rem:2}, \rref{eq:pr:th2:absx}, and \rref{eq:pr:th2:absG}, we conclude that, for each $l \in \llbracket 1 , j+1  \rrbracket$ and each $a \in \llbracket 1 , l \rrbracket $, there exists  a multivariate polynomial $\overline{P}_{l,a} : \reels^n \to \reels $ of degree $2a+1$, which not depend on $\beta$, such that, for any $t \geq 0$,
\begin{equation}
\label{eq:pr:th2:absbell}
\abs{B_{l,a}( G^{(1)}(t), \ldots , G^{(l-a+1)}(t) ) b^T x^{(k-l)}(t) } \leq \overline{P}_{l,a} (\abs{x_1(t)}, \ldots , \abs{x_2(t)}).
\end{equation}
Since $\bar P_{l,a}$ and $P_{i,j+1}$ are respectively of degree $2a+1$ and $1$ and recalling that $\alpha\geq 1/2$, we conclude that, for each $l \in \llbracket 1 , j+1  \rrbracket$ and each $a \in \llbracket 1 , l \rrbracket $, there exists a positive constant $M_{l,a,j}$ such that
\begin{equation}
\label{eq:pr:th2:exp1}
\sup \lbrace \abs{\overline{P}_{l,a} (x)/(1 + \norme{x}^2)^{ \alpha +a }} : x \in \reels^n \rbrace \leq M_{l,a,j},
\end{equation}
and, for each  $i \in \llbracket 1 , n \rrbracket$, there exists a positive constant $Q_{i,j}$ such that
\begin{equation}
\label{eq:pr:th2:exp2}
\sup \lbrace \abs{P_{i,j+1} (x)/(1 + \norme{x}^2)^{ \alpha}} :  x \in \reels^n  \rbrace \leq Q_{i,j}.
\end{equation}

Let $\beta \leq \beta_j$ and $x(\cdot)$ be a trajectory of $\dot{x}=Ax+b\nu(x)$. Thanks to Assertion \ref{assert2} (with $k=j+1$), we get that, for all $t\in\reels_{\geq 0}$, $\abs{U^{(j+1)}(t)}  \leq \xi (t)$ with
\begin{align*}
\xi(t) = \beta \Big( \frac{ \abs{b^{T} x^{(j+1)}(t)}}{(1 + \norme{x(t)}^2)^\alpha } +  \sum\limits_{l=1}^{j+1} \binom{j+1}{l} \sum\limits_{a=1}^l d_a \frac{\abs{ B_{l,a}(\underline{G}(t)) b^T x^{(j+1-l)}(t)} }{(1 + \norme{x(t)}^2)^{\alpha + a}} \Big),
\end{align*} where $B_{l,a}(\underline{G}(t))=B_{l,a}(G^{(1)}(t) , \ldots , G^{(l-a+1)}(t)) $. Therefore a straightforward computation using \rref{eq:pr:th2:absx}, \rref{eq:pr:th2:absG}, \rref{eq:pr:th2:absbell}, \rref{eq:pr:th2:exp1}, and \rref{eq:pr:th2:exp2} leads to the existence of a positive constant $M_j$ such that $\abs{U^{(j+1)}(t)} \leq \beta M_j$, for all positive time. Thus for all $\beta \leq \min \lbrace \beta_j, \underline{R}/M_j \rbrace$, it follows that $\sup \lbrace \abs{U^{(j+1)}(t)} : t\geq 0 \rbrace \leq \underline{R}$. This ends the induction on $j$ and concludes the proof of Theorem \ref{theo2}.

\section{Numerical examples}
\label{sec:simu11}
\subsection{The triple integrator}
In this subsection, we illustrate the applicability and the performance of the feedback law proposed for Case 1 on a particular example. We use the procedure described in Section \ref{sec:proof_th} in order to compute a $2$-bounded feedback law by $(2,20,18)$ for the multiple integrator of length three. Our set of saturation functions is $\sigma_1 = \sigma_2 = \sigma_3 = \sigma$, where $\sigma$ is an $\scal(2)$ saturation function with constants $(2,1,2,1)$ given by $r$ if $ \abs{r} \leq 1$, $\sign (r) ( -4 + 15 \abs{r} - 18 r^2 + 10 \abs{r}^3 - 2r^4 )$ when $1 \leq \abs{r} \leq 1.5$, $2 \sign (r) ( 25 - 60 \abs{r} + 54 r^2 -21 \abs{r}^3 + 3 r^4 )$ if  $1.5 \leq \abs{r} \leq 2$, and $2  \sign (r)$ otherwise.
We choose $\mu_2^{max} = 2/5$, $L_{\mu_2} = 1/5$, $\mu_1^{max} = 1/12$, and $L_{\mu_1} = 1/24$. Following the procedure, we obtain that the two first time derivatives of the control signal $U(\cdot)=\Upsilon(y(\cdot))=\Upsilon(Hx(\cdot))$, with $H=( h_1 , h_2 , h_3)$ where $h_1=(1 / \lambda^2, \:  2 / \lambda , \: 1)^T  $, $h_2= (0 , \: 1/ \lambda ,\: 1)^T$ and $ h_3=e_3$, satisfy $\sup \lbrace \abs{u^{(1)}(t)} : t\geq 0 \rbrace  \leq 7.91 / \lambda^2 + 4.35 / \lambda$ and $\sup \lbrace \abs{u^{(2 )}(t)} : t\geq 0 \rbrace  \leq 26.2/ \lambda + 396 / \lambda^2 + 1147.2 / \lambda^3 +125.2 / \lambda^4$.

One get that $\sup \lbrace \abs{U^{(1)}(t)} : t\geq 0 \rbrace  \leq 0.9 $, and  $\sup \lbrace \abs{U^{(2)}(t)} : t\geq 0 \rbrace  \leq 18 $ by choosing $\lambda=6.5$. Observing that the amplitude of $U$ is below $2$ by construction, this confirms the fact that this is a $2$-bounded feedback by $(2,20,18)$. The desired feedback is then given by
\begin{align*}
\nu(x) = &  - \sigma \Big( \frac{1}{6.5} \Big( x_3 + \frac{1}{5} \sigma \big(5 ( h_2^T x + \frac{1}{24} \sigma \big( 24 (h_1^T x )) \big)\big)\Big)\Big).
\end{align*}

This feedback law is tested in simulations and the results are presented in Figure \ref{figub}. Trajectories of triple integrator with the above feedback are plotted in grey for several initial conditions. The corresponding values of the control law and its time derivatives up to order $2$ are shown in Figure \ref{figub}. These grey curves validate the fact that asymptotic stability is reached and that the control feedback magnitude, and two first derivatives, never overpass the prescribed values $(2,20,18)$. In order to illustrate the behaviour of one particular trajectory, the specific simulation obtained for initial condition  $x_{10}=446.7937$, $x_{20} = -69.875$ and $x_{30}=11.05$ is highlighted in bold black. It can be seen from Figure \ref{figub} that our procedure shows some conservativeness as the amplitude of the second derivative of the feedback never exceeds the value $2$, although maximum value of $18$ was tolerated.

\begin{figure}[thpb]
     \centering
     \includegraphics[height=8cm, width=\textwidth]{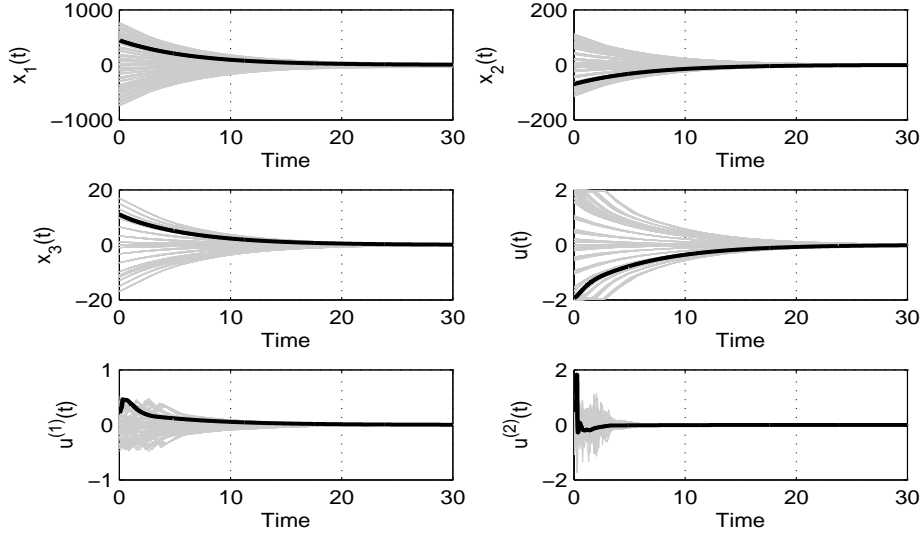}
     \caption{Evolution of the states, the control and its derivative up to order $2$ for a set of initial conditions.}
    \label{figub}
\end{figure}
\subsection{The harmonic oscillator} \label{sec:simu}
We finally test the performance of the control law proposed for Case 2 through example of a $1$-bounded feedback law by $(2,2)$ for an harmonic oscillator. We consider the following system $ \dot{x}_1 =  5 x_2$, $\dot{x}_{2}  = - 5 x_1 + u$.

In accordance with Theorem \ref{theo2}, we take $u(x_1,x_2)= - \beta x_2/\sqrt{1 + x_1^2 + x_2^2}$ with $\beta= (-5 + \sqrt{41})/{4}$. The behaviour of the resulting closed-loop system and the corresponding values of the feedback and its first time derivative are shown in Figure \ref{figud} for initial conditions $x_{10}=2$ and $x_{20}=-2$. It can be seen that the values of the control and its time derivative stay below $2$ as desired.

\begin{figure}[thpb]
      \centering
      \includegraphics[height=6cm, width=\textwidth]{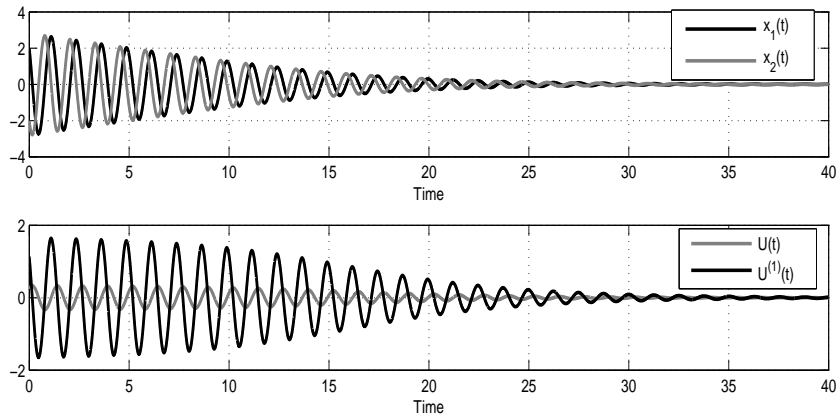}
      \caption{Evolution of the states, the control and its first derivative for initial conditions $x_{10}=2$ and $x_{20}=-2$.} 
      \label{figud}
\end{figure}

\bibliographystyle{elsarticle-num} 
\bibliography{biblio}

\end{document}